%% file: FF_MUBs.tex
\documentclass[a4paper]{amsart}

\input{defns}

\newcommand{\mdf}[2]{\mathcal{M}_{#1}{#2}}

\begin{document}
\title{Evidence for and against Zauner's MUB Conjecture in $\CC^6$}
\author{Gary McConnell}
\address{Controlled Quantum Dynamics Theory Group, Imperial College, London}
\email{g.mcconnell@ic.ac.uk}
\author{Harry Spencer}
\address{Churchill College, University of Cambridge}
\email{hs696@cam.ac.uk}
\author{Afaq Tahir}
\address{Pembroke College, University of Oxford}
\email{afaq.tahir@pmb.ox.ac.uk}

\begin{abstract}
The problem of finding provably maximal sets of mutually unbiased bases in $\CC^d$, for composite dimensions $d$ which are not prime powers, remains completely open. 
In the first interesting case,~$d=6$, Zauner predicted that there can exist no more than three MUBs.  

We explore possible algebraic solutions in~$d=6$ by looking at their~`shadows' in vector spaces over finite fields. 
The main result is that if a counter-example to Zauner's conjecture were to exist, then it would leave no such shadow upon reduction modulo several different primes, forcing its algebraic complexity level to be much higher than that of current well-known examples. 

In the case of prime powers~$q \equiv 5 \bmod 12$, however, we are able to show some curious evidence which --- at least formally --- points in the opposite direction. 
In $\CC^6$, not even a single vector has ever been found which is mutually unbiased to a set of three MUBs. 
Yet in these finite fields we find sets of three `generalised MUBs' together with an orthonormal set of four vectors of a putative fourth MUB, all of which lifts naturally to a number field. 
\end{abstract}

\maketitle

\section*{Introduction}
The notion of \emph{mutually unbiased bases} or \emph{MUBs} arose in physics as an optimal choice of measurement bases for quantum tomography~\cite{schwinger, ivanovic, wootters, zauner}; although the concept was discovered independently in combinatorial design theory~\cite{camseid,CCKS,godsilroy}. 
The problem of finding provably maximal sets of MUBs in $\CC^d$ for non-prime-power dimensions $d$ remains completely open. 
In the first interesting case,~$d=6$, Zauner~\cite{zauner} has predicted that there can exist no more than three MUBs. 
Moreover there is a conjecture about orthogonal decompositions of Lie algebras~\cite{kost} which is equivalent to saying that for any such non-prime-power~$d$, a maximal set of $d+1$ MUBs cannot exist: see~\cite{ingemar3}. 

A priori there is no reason to expect the entries of MUB vectors over~$\CC$ to be algebraic. 
Indeed in~$\CC^6$ there is a catalogue of continuously parametrised families~\cite{jaming,szoll} of sets of three MUBs which outside a set of measure~$0$ are not unitarily equivalent to any algebraic set; showing that for arbitrary MUBs, the algebraic complexity level of the entries of the vectors can go all the way up to the transcendental. 
However since the equations are defined over~$\ZZ$, one approach to exploring the known bounds on the number of MUBs over~$\CC^d$ is to attack the problem as an algebraic question over a general ring, and to see whether the equations can be solved there.

When studying finite systems of equations with integer coefficients, a standard technique in number theory and algebraic geometry is to look at their reductions modulo a prime number. 
Nambu~\cite{nambu}, for example, also applied the same technique in a broad range of physical problems. 
Simplifying slightly for clarity, one looks at the image modulo~$p$ of their integer solutions, `most' of which will survive reduction to the field~$\FF_p$ of integers modulo~$p$. 
This picture naturally extends to number fields, and consequently the \emph{absence} of a solution in a finite field can under strict conditions be used to show a corresponding failure of solvability back up in the original number field. 
See~\cite{grief} for a recent similar approach to an affiliated problem.

The main thrust of this paper is to add a small result arising from this hitherto unexplored direction, to the already large body of evidence adduced in favour of Zauner's conjecture. 
Namely, in theorem~\ref{nogo} we demonstrate in a number of finite fields that sets of MUBs of size~$4$ in dimension~$6$ do {not} exist. 
{Prima facie} this is weak evidence that such sets of MUBs do not exist in the sorts of small degree number fields in which solutions have previously been studied: see for example~\cite{ingemar6,dardo,appingdang} and~\ref{redux}. 
Nevertheless, the translation back up to characteristic zero is still computationally intractable, so applying this result relies on studying each specific situation individually. 

In the case of prime powers~$q\equiv5 \bmod 12$, however, proposition~\ref{quiver} points tentatively in the opposite direction. 
We lift a set of three generalised MUBs plus an additional set of four vectors of a putative fourth MUB, directly to solutions in a number field. 
Note however that we are beginning with a set of three what one might deem \emph{hyperbolic MUBs}: ones whose entries are all \emph{totally real} numbers and where the property of unitarity is measured in terms of a real (algebraic) norm as opposed to a complex absolute value.

It is important to note that when~$F$ is finite or~$p$-adic, the sesquilinear form~$\langle\ ,\ \rangle$ yields a weaker geometric structure than an inner product~\cite{grove}. 
In particular,~$F^d$ will always contain isotropic vectors.  
Moreover, searches based on local minima of analytical functions --- which have been deployed in most of the attempts to solve this problem over the complex field, for example~\cite{debz, stefan, jaming, raynal, dardo, grassmc} --- are not possible at the finite field level. 
So there is no bijective geometric correspondence between our findings in arbitrary finite or~$p$-adic fields, and the solutions in Hilbert spaces.  

In the appendices we provide very brief notes on the results of the computer searches, the equations and their reductions, and general background reference material on the properties of these generalised MUBs.

\section{The MUB problem over~$\CC$ and beyond}
\subsection{State of knowledge over~$\CC$}\label{statknow}
The MUB problem over~$\CC$ has been attacked using a mixture of geometric, combinatorial and numerical methods.  
To some extent these approaches have been unified under the combinatorial umbrella of frames and complex projective~$2$-designs: see e.g.~\cite{godsilroy,CCKS,scottroy, grossmcc}. 
We refer to the vast body of work in the mathematical physics literature on this problem --- for example~\cite{schwinger, ingemar3, ingemar6, chterental, debz, babo} --- for the motivation for studying Hilbert space MUBs and for various standard results about their geometry. 

Before giving a generalised definition of MUBs in the next section, we set out the current state of knowledge on maximal sets of MUBs in complex Hilbert space~$\CC^d$. 
Let~$\mdf{d}{\CC}$ denote the maximum number of orthonormal bases of~$\CC^d$ which can be pairwise mutually unbiased to one another.

\begin{theorem*}\rm\cite{alltop, ivanovic, wootters, klaro}\it\ \ 
For~$d\geq2$, write $d=\prod_{i=1}^np_i^{r_i}$ as a product of its prime power factors with the $p_i$ ordered so that $p_1^{r_1}<p_2^{r_2}<\ldots<p_n^{r_n}$. 
Then: 
\begin{enumerate}[\rm (I)]
\item{$
p_1^{r_1}+1\ \leq\ \mdf{d}{\CC}\ \leq\  d+1.
$}

\item{When $d$ is a prime power, a maximal set of $d+1$ MUBs always exists in~$\CC^d$. }
\end{enumerate}
\end{theorem*}
The lower bound, which in the general case is currently limited to the outcome of taking tensor products~\cite{zauner,babo,asch,klaro}, has in fact been marginally improved~\cite{wocjan} in some non-prime-power squared dimensions like~$d = 676$. 
It should be stressed again that no-one yet knows tight lower or upper bounds for~$\mdf{d}{\CC}$ for any~$d$ not a power of a prime.

\subsection{Generalised definition of MUBs}\label{gendef}
A formal definition of mutual unbiasedness may be made by analogy with the motivating complex case. 
For other instances of this see for example~\cite{chterental,boyband,vandam}\footnote{Although ostensibly looking at the same question,~\cite{vandam} was concerned with infinite-dimensional Hilbert spaces. The results and techniques are completely unrelated to those of this paper. }; and for further discussion of our case see~\ref{axcee} and~\ref{axceetoo}. 
Throughout,~$d\in\NN$ will denote the dimension of our vector space. 
Let~$F$ be any field of characteristic not dividing~$2d$, with an automorphic involution~$\sigma$ whose fixed subfield is~$K$. 
So in the usual complex case, $F=\CC$ and $K=\RR$ and~$\sigma$ is complex conjugation. 
For any~$\alpha\in F$,~$\norm{F}{K}{\alpha} =  \alpha\alpha^\sigma$ is the field norm down to~$K$.

By~$F^d$ we shall always mean a~\emph{unitary space} \cite[chapter 10]{grove}: namely, a~$d$-dimensional vector space over~$F$ equipped with a non-degenerate \emph{Hermitian} form defined for any pair of vectors~$\mathbf{u} = (u_j), \mathbf{v} = (v_j) \in F^d$ by: 
$$\langle {\mathbf{u}}, {\mathbf{v}}\rangle =  \mathbf{u}^\dag \cdot \mathbf{v}  =  \sum_{j=1}^{d}u_j^\sigma v_j, 
$$
where for any matrix~$M$ with entries in~$F$, $M^\dag$ denotes its conjugate transpose with respect to the involution~$\sigma$. 
The word \emph{basis} will always mean~\emph{orthonormal basis}, with respect to the given form; so we may write the basis set~${\mathcal{B}} = \{ {\mathbf{b}}_1,\ldots,{\mathbf{b}}_d \}$ as the column vectors of a \emph{unitary} --- as defined with respect to this Hermitian form --- matrix~$[ {\mathbf{b}}_1,\ldots,{\mathbf{b}}_d ]$. 

\begin{definition}\label{mubdef}
Let~$\mathbf{u} , \mathbf{v} \in F^d$ be any choice of~\emph{unit vectors}: that is, satisfying~$\langle \mathbf{u} , \mathbf{u} \rangle = \langle \mathbf{v} , \mathbf{v} \rangle = 1$. 
We say that~$\mathbf{u} , \mathbf{v}$ are \emph{mutually unbiased} or~\emph{MU} to one another if
\begin{equation}\label{mubness}
\norm{F}{K}{ \langle \mathbf{u} , \mathbf{v} \rangle } \ =\ 1 / d . 
\end{equation}
\noindent 
Let~${\mathcal{B}}$ and~${\mathcal{C}}$ be any two orthonormal bases. 
If~$\norm{F}{K}{ \langle \mathbf{b}_i , \mathbf{c}_j \rangle } = 1 / d$ for every pair of vectors~$ \mathbf{b}_i \in \mathcal{B}$, $\mathbf{c}_j \in \mathcal{C}$ then we say that~${\mathcal{B}},{\mathcal{C}}$ are \emph{mutually unbiased bases}. 

We use the same notation as for the complex case, in that~$\mdf{d}{F}$ will denote the size of a maximal set of MUBs in~$F^d$. 
\end{definition}
The definition is motivated by the situation in~$\CC^d$, in that if we require that {every} inner product~$\langle {\mathbf{b}}_i, {\mathbf{c}}_j\rangle$ for $1\leq i,j\leq d$ have the same absolute value, then in fact that common value is readily proven~\cite{schwinger} to be~$1/\sqrt{d}$. 

It is important to note that in translating \emph{particular} known algebraic solutions directly from number fields embedded in~$\CC$, down to finite and~$p$-adic fields, at least half of the time we encounter a situation where the~$q$-th power Frobenius automorphic involution~$\sigma$ does not act analogously to complex conjugation on the roots of unity. 
This means that the Wootters \& Fields (henceforth just WF) solutions in~\cite{wootters} to the equations in~$\CC$ simply \emph{do not exist} modulo~$p$. 
This is evidenced for example in the number of primes implicitly excluded from the hypotheses of proposition~\ref{mdffpp}, and even more prominently in the tabulation of our exhaustive search results in~\ref{tableau}. 

Conversely, any finite field solutions are by definition cyclotomic numbers; so this constraint is inescapable. 
Using a set of equations derived from the complex situation is thus rendered meaningless. 

Throughout the paper, therefore, we have stuck with the \emph{geometric} interpretation of the MUB existence question, adjusting the equations so that in each finite field they reflect the notion of an adjoint or Hermitian dot product, rather than interpreting it strictly as the specialisation of a complex algebraic variety.

The level of overdetermination in this problem is formidable. 
Just in the relatively tight case in proposition~\ref{quiver}, even after all sensible reductions, we still begin with~$150$ variables over~$\FF_q$ but~$231$ equations. 
Elimination-theoretic tools like Gr\"obner bases had to be rejected in favour of an exhaustive search through small finite vector spaces. 
We have used~MAGMA software~\cite{magma} throughout.

\section{MUBs over finite fields : implications for Zauner's conjecture}\label{nonarx}
The link between the study of complex MUBs and those over quadratic extensions of finite fields is given explicitly by the correspondence between the canonical complex conjugation involution whose fixed field is~$\RR$, and the~$q$-th power \emph{Frobenius} involution which is the unique non-trivial automorphism of the field~$F$ fixing the base field~$K$. 
It is this structural parallel which affords the possibility of analogous geometric results in the two contexts. 
This is explained further in~\ref{drownsec}. 

Our attempts to explore the lower bound for~$\mdf{6}{\CC}$ via the images of algebraic MUBs under reduction in finite fields, have led us to the following.

\begin{theorem}\label{nogo}
~$\mdf{6}{\FF_{q^2}} \leq 3$ for every prime power~$q = p^r$ in the set
\[
\{ 5 , 25, 7 , 49 , 11 , 13 , 17 , 19 , 23, 29, 31, 37, 41 \}.  
\]
\end{theorem}

\begin{proof}
The evidence from the numerical searches is tabulated in~\ref{tableau}. 
The methodology is explained in~\ref{simp}. 
\end{proof}

Theorem~\ref{nogo} supports Zauner's conjecture in a way which has not been explored before. 
However, as mentioned in the introduction, we also have a result which gives some slight evidence in the opposite direction. 

\begin{proposition}\label{quiver}
Let~$F = \FF_{q^2}$ for some~$q \equiv 5 \bmod 12$. 
Then there exists in~$F^6$ a set~$\mathcal{S}_F  =  \{ \mathcal{B}_0, \mathcal{B}_1, \mathcal{B}_2 \}$ of three pairwise mutually unbiased bases, together with an orthonormal set of four vectors which are all mutually unbiased to the three bases. 
In each case the solutions lift to the number field extension~$K(\sqrt{3}) /  K$ where~$K = \QQ(i,\sqrt{6})$. 
\end{proposition}

\begin{proof}
See~\ref{quills} where the calculations are explained. 
Here, as in the rest of the paper, the initial basis set~$\mathcal{B}_0$ is chosen to be the \emph{computational basis} consisting of the $d$ vectors $\mathbf{e}_j = (0,\ldots,0,1,0,\ldots,0),$
where the unique non-zero entry occurs at the~$j$-th position. 
$\mathcal{B}_0$ is represented in matrix form by the identity matrix. 
An intermediate step was to lift the original solutions from the searches over finite fields, using standard~$p$-adic techniques, up several levels in order to find the defining equations for the entries. 
It then became clear that they lifted not only to unramified extensions~$F_\wp$ of~$\QQ_p$ but in fact all the way to a number field contained in~$F_\wp$. 
It is also important to note that there are infinitely many possible number fields to which such solutions may be lifted: we have just chosen the simplest. 
\end{proof}

As far as we can determine from the published literature, the arXiv, and personal communications with several other researchers active in this field (thanks to Ingemar Bengtsson, Stefan Weigert, Markus Grassl, Dan McNulty, Marcus Appleby and Dardo Goyeneche), no-one has yet been able to find a \emph{single vector} in~$\CC^6$ which is mutually unbiased to a known set of three MUBs~\cite{dardo, stefandan}. 
Indeed, using Gr\"obner basis techniques for varying values of~$d$, Markus Grassl et al.~have demonstrated many cases wherein the extension of certain \emph{non-maximal} sets of bases in~$\CC^d$ \emph{by just a single vector} is impossible, elucidating a bewilderingly rich variety of possible outcomes given simple variations in the starting bases~\cite[\S III]{grassmc},~\cite{desig}. 

So the existence of this extra set of four vectors in proposition~\ref{quiver}, even if it only obtains in this artificial `hyperbolic' context, raises a natural question as to whether such a phenomenon occurs as some sort of Galois image of an actual set of MUBs, analogously for example to \emph{ghost SICs}~\cite{ghosts} in the cousin \emph{SIC-POVM} problem~\cite{ing}.

\section{Values of~$\mdf{d}{F}$ when~$F$ is a finite field}
The best bounds which may be stated in full generality are as follows. 

\begin{proposition}\label{finivan}
For any quadratic extension of finite fields~$F/K$ as above and any dimension~$d\geq2$ coprime to the characteristic~$p$ of~$K$,
$$1 \leq \mdf{d}{F} \leq d+1 . $$
\end{proposition}

\begin{proof}
Despite there always being at least three MUBs over~$\CC^d$ for any~$d\geq 2$ --- which~\cite{babo} may be regarded simply as the eigenbases of the generalised Pauli operators~$X,Z$ and their product~$XZ$ --- the lower bound in this finite field case cannot be improved upon in general. 
For example, see table~\ref{teatoo} in~\ref{tableau} in which~$\mdf{3}{\FF_{q^2}}$ is shown to be~$1$ for essentially half of all~$q$.

For the upper bound, which again is saturated in many cases, we adapt the argument from the complex case in~\cite{ivanovic}, following the exposition in~\S4 of~\cite{ingemar3}. 

\def\mddf{{\hbox{\rm M}_{d\times d}(F)}}

Define a non-degenerate Hermitian trace form~$\hbox{\rm Tr}\left( A^\dag B \right)$ on the matrix algebra~$\mddf$ of degree~$d$ over~$F$. 
Consider any orthonormal basis~$\mathcal{B} = \{\mathbf{v}_1 , \ldots , \mathbf{v}_d \}$ of~$F^d$. 
The corresponding rank~$1$ trace~$1$ projectors~$\pi_\mathbf{v_k}  =  \mathbf{v}_k \otimes \mathbf{v}_k^\dag$ are orthogonal unit vectors in~$\mddf$ under this trace form. 
Their endpoints span a~$d$-simplex whose barycentre is~$\frac{1}{d}\II$, where~$\II$ denotes the identity matrix in~$\mddf$. 
The set resulting from subtracting~$\frac{1}{d}\II$ from each~$\pi_\mathbf{v_k}$ spans a~$K^{d-1}$ subspace of the space~$\mathcal{H}_0 \cong K^{d^2-1}$ of trace~$0$ Hermitian operators in~$\mddf$. 
The pairwise Hermitian products between distinct shifted projectors arising from~$\mathcal{B}$ are now~$\frac{-1}{d}$ rather than zero.

However, given any~$\mathbf{b} \in \mathcal{B}$ and~$\mathbf{c} \in \mathcal{C}$ from mutually unbiased bases~$\mathcal{B}$ and~$\mathcal{C}$, the shifted trace product is zero: 
\Small
\begin{eqnarray*}
\hbox{\rm Tr}\big((\mathbf{b} \otimes \mathbf{b}^\dag - \frac{1}{d}\II )^\dag (\mathbf{c} \otimes \mathbf{c}^\dag  - \frac{1}{d}\II )\big)  &  =  &  \hbox{\rm Tr} \big( (\mathbf{b} \otimes \mathbf{b}^\dag)(\mathbf{c} \otimes \mathbf{c}^\dag)\big) - 
\frac{1}{d}\hbox{\rm Tr}\big(  (\mathbf{b} \otimes \mathbf{b}^\dag  + \mathbf{c} \otimes \mathbf{c}^\dag )\big)  + \hbox{\rm Tr} \frac{1}{d^2}\II  \\
&  =  &  \langle \mathbf{b} , \mathbf{c} \rangle  \langle \mathbf{c} , \mathbf{b} \rangle  -  \frac{2}{d}  + \frac{1}{d} \\
&  =  &  0.
\end{eqnarray*}
\normalsize
That is to say, mutual unbiasedness between the bases down in~$F^d$ lifts to orthogonality between the corresponding translated operator subspaces in~$\mathcal{H}_0$. 

Hence~$n$ MUBs of~$F^d$ produce~$n$ mutually orthogonal $(d-1)$-dimensional $K$-subspaces inside~$K^{d^2-1}$. 
If~$p\mid {\rm char}(K)$ then for example the identity matrix lies in~$\mathcal{H}_0$ and its trace product with every other matrix in~$\mathcal{H}_0$ is zero, rendering degenerate the~$K$-valued restriction to~$\mathcal{H}_0$ of this trace form. 
However when~$p \nmid {\rm char}(K)$, by considering a $K$-basis for~$\mathcal{H}_0$ formed from simple $F$-linear combinations of the elementary matrices~$E_{ij}$ we see that the trace form remains non-degenerate; and the result follows. 
\end{proof}

\begin{remark}
The diagonalisation-based argument of~\cite{babo} does not work here, because no Hermitian structure is possible. 
Indeed, the Galois group $\Gal{\overline{F}}{F} \cong \widehat{\ZZ}$ of the algebraic closure $\overline{F}$ of a finite field $F$ is torsion-free. 
So in particular there is no surrogate operator for complex conjugation to match that of $\Gal{\CC}{\RR}$. 
\end{remark}

\subsection{Values of~$\mdf{d}{\FF_{q^2}}$ when~$d$ is a prime power}\label{fdpp}
Given any prime power dimension~$d=l^k$, the next result yields a class of primes~$p$ of positive density, by Dirichlet's theorem on primes in arithmetic progressions, for which the complex WF solutions have \emph{good reduction} over fields of characteristic~$p$. 

\begin{proposition}\label{mdffpp}
Let~$l$ be a prime and~$k$ a positive integer. 
Let the prime $p\neq l$ and~$r\geq1$ satisfy~$p^r\equiv-1\bmod l^k$ when~$l$ is odd, or simply~$p^r \equiv -1 \bmod 4$ when $l=2$. 
Then writing~$q = p^r$, 
$$\mdf{l^k}{\FF_{q^2}}  =  {l^k+1}.$$ 
\end{proposition}

Notice that if~$p^r \equiv -1 \bmod l^k$ then the same is true of~$p^{ar}$ for every odd integer~$a$. 
In other words, for each particular dimension~$d=l^k$, every such prime~$p$ additionally furnishes us with an infinite set of finite fields for which the statement holds. 

The same techniques can be applied to more complicated sets of MUBs: we illustrate this in~\ref{drownsec}. 
Moreover, as illustrated in~\ref{tableau}, whenever we have done an exhaustive search in a prime-power dimension~$d=2,3,4,5,7$ over some finite field, we have found that all \emph{maximal} solutions of~$d+1$ MUBs are \emph{equivalent} --- in the sense of~\cite{BWB} as explained in~\ref{axceetoo} --- to the reduced WF solutions, possibly modified by~$\sqrt{-1}$ as in the proof below. 

On the other hand, for prime powers where the structure in proposition~\ref{mdffpp} does not obtain, we have shown that for~$d\leq7$,~$d+1$ MUBs do \emph{not} exist for the small primes in our searches\footnote{Other than in the strange anomalous case~$d=7$,~$p=3$ in table~\ref{sebben} of~\ref{tableau}, where \emph{a priori} we only know that solutions exist for powers of~$q=p^3 = 27$, by proposition~\ref{mdffpp}. }. 
So on the basis of this tiny body of evidence, the WF construction --- with the slight modifications needed in the proof below --- would seem to be universal for creating maximal sets of MUBs over finite fields. 
See~\cite{babo} for a comparison with the complex case. 

The sub-maximal sets in these defective dimension-prime combinations display behaviour akin to \emph{bad reduction} of abelian varieties; although it is not because of the introduction of mod-$p$ singularities but rather because the Hermitian form collapses. 
In other words, as we observed in~\ref{gendef}, because the Hermitian form is encapsulated (via the Galois action) within the polynomial system, we are forced to use a different set of `real' equations in these finite fields from that in the complex case or indeed in the finite fields in proposition~\ref{mdffpp}. 
It is not clear what the `geometric' meaning of the solutions is when there is a square root of~$-1$ in the ground field, as is the case for example in proposition~\ref{quiver}.

Given any field~$F$ and any~$N\in\NN$, let~$\zeta_N\in\overline{F}$ denote a fixed primitive~$N$-th root of unity in an algebraic closure~$\overline{F}$ of~$F$ and let~$\mu_N\ =\ <\zeta_N>$ be the group it generates. 
Similarly the symbol~$\sqrt{T}$ will denote a fixed square root of a field element~$T$ inside an algebraic closure. 
A reference for the facts we use on finite fields is~\cite{serre}. 

\begin{proof}[Proof of proposition~\ref{mdffpp}]
We show by construction that the upper bound is attained; that it cannot be breached is proposition~\ref{finivan}. 
In~\cite{wootters}, complete sets of~$l^k+1$ MUBs $\{\mathcal{B}_0,\mathcal{B}_1 ,\ldots,,\mathcal{B}_{l^k} \}$ are constructed in~$\CC^{l^k}$, where~$\mathcal{B}_0$ is the {computational basis} represented by the identity matrix. 
All vector entries in the other~$\mathcal{B}_k$ are of the form~$\frac{\zeta_l^t}{\sqrt{l^k}}$ for some~$t \in \ZZ$ when~$l$ is odd; and~$\frac{i^t}{\sqrt{2^k}}$ when~$l=2$, where~$i=\sqrt{-1}$. 

Suppose first that~$l$ is an odd prime. 
The cyclic Galois group~$\Gal{\QQ(\zeta_l)}{\QQ}$ has a unique involution~$\tau$ whose inversion action on~$\mu_l$ is identical to that of the restriction of complex conjugation. 
On the other hand~$\sqrt{l}$ is real and so it is fixed under complex conjugation. 
It follows that the conjugation action on the vector entries in~$\{\mathcal{B}_0,\mathcal{B}_1 ,\ldots,,\mathcal{B}_{l^k} \}$ is captured entirely by~$\tau \colon \zeta_l \mapsto \zeta_l^{-1}$. 

We have assumed that~$q \equiv -1 \bmod l^k$ and~$l\neq2$, hence $\FF^\times_q  \cong  \mu_{q-1}$ contains no~$l$-th roots of unity other than~$1$; whereas~$\FF^\times_{q^2}  \supset \mu_{q+1} \supset \mu_{l^k} \supset \mu_l$. 
In particular,~$\FF_{q^2}  =  \FF_q(\zeta_l)$ and the $q$-th power Frobenius automorphism~$\sigma$ in the extension~$\FF_{q^2} / \FF_q$ acts on the~$l$-th roots of unity via inversion, since~$\zeta_l^\sigma = \zeta_l^q = \zeta_l^{-1}$. 
It follows that any solution to the MUB equations in~$\frac{1}{\sqrt{l^k}}\ZZ(\zeta_l)$ will satisfy the very same equations down in~$ \FF_{q^2}$, {provided that}~$\sigma$ fixes the square root of~$l^k$. 

If~$k$ is even then~$l^k$ is a square integer; and if~$r$ is even then by the uniqueness of the quadratic extension of~$\FF_p$ the square root of~$l\in\FF_p$ must already be contained in~$\FF_q$. 
In such cases the norm~$\norm{\FF_{q^2}}{\FF_q}{\sqrt{l^k}}  =  l^k$ and consequently the behaviour exactly mimics that in the complex case.  

Note also the trivial fact that when~$p=2$, so~$l$ is in fact forced to be odd, the square root of~$l^k$ --- that is,~$1$ --- is automatically in the ground field. 

However if all of~$p$,~$l$,~$k$ and~$r$ are odd it is possible that~$l^k$ will not be a square in the base field~$\FF_q$: hence~$\sigma$ maps~$\sqrt{l^k} \mapsto -\sqrt{l^k}$. 
In such cases every instance of~$\sqrt{l^k}$ in the WF example must be twisted by some element~$\nu \in \FF_{q^2}$ whose norm is~$-1$. 
But the norm map from~$\FF_{q^2}^\times$ to~$\FF_q^\times$ is surjective and so there are exactly~$q+1$ elements of~$\FF_{q^2}^\times$ which map onto~$-1$. 
Choose one such~$\nu \in \FF_{q^2}^\times$ and multiply all of the WF vector entries by~$\nu$, ignoring~$\mathcal{B}_0$ of course; then once again~$\norm{\FF_{q^2}}{\FF_q}{\nu\sqrt{l^k}}  =  l^k$ and the set of equations derived from the geometry in~$\CC^{l^k}$ is satisfied.

When~$l=2$, so~$p$ is odd, all of the coefficients of the WF MUBs in dimension~$2^k$ are of the form~$\frac{i^a}{\sqrt{2^k}}$ for some~$a\in\ZZ$. 
In a finite field of odd characteristic, we need the Frobenius action~$\sigma$ on the fourth roots of unity and on~$\sqrt{2}$, where relevant, to mimic that in~$\CC/\RR$: namely, as inversion on~$i = \sqrt{-1}$ and as the identity on~$\sqrt{2}$. 
This rules out~$q \equiv 1 \bmod 4$, since then either~$p \equiv 1 \bmod 4$ or~$r$ is even: either or both of which would ensure that~$-1$ is a square in~$\FF_q$. 
So we are forced into the situation in the hypotheses of the proposition, where~$p \equiv 3 \bmod 4$ and~$r$ is odd, which is equivalent to~$q \equiv 3 \bmod 4$. 

Hence all remaining possibly problematic cases for~$l=2$ reduce to~$k$ odd,~$r$ odd,~$p \equiv 3 \bmod 8$ or~$p \equiv 7 \bmod 8$.
When~$p \equiv 7 \bmod 8$,~$\leg{2}{p} = 1$ and so the WF complex solutions behave identically down in~$\FF_{q^2}$. 

When~$p \equiv 3 \bmod 8$, however, the Legendre symbols~$\leg{-1}{p} = \leg{2}{p} = -1$ and so both~$\sqrt{-1}$ and~$\sqrt{2}$ are mapped by~$\sigma$ to their negatives. 
The net action of~$\sigma$ is therefore to fix entries of the form~$\frac{\pm i}{\sqrt{2}}$ and to flip the signs on the `real' entries~$\frac{\pm1}{\sqrt{2}}$. 
So we need once again to compensate by adding in a factor whose norm is~$-1$, and we follow the same procedure as above. 
\end{proof}

\begin{remark}
Under the~\emph{Alltop construction}~\cite{alltop,klaro,blanch} the complete set of~$d+1$ MUBs may also be constructed\footnote{For odd primes~$l$ or~$d=2,4$; however this is apparently still an open problem when~$d = 2^k$ for~$k\geq 3$. See~\cite{blanch} for the case~$d=4$ and for the history of the problem. } as the images of a~\emph{fiducial vector} under the action of a \emph{Weyl-Heisenberg} or \emph{generalised Pauli} group of degree~$d$, just as in the cousin problem of SIC-POVMs~\cite{ing}. 
As in the proof above, all of the matrix entries as well as those of the fiducial vector are drawn from~$\frac{1}{\sqrt{l^k}}\ZZ(\zeta_l)$; and the only group action relevant to the MUB properties is that of complex conjugation. 
Hence with similar modifications, the construction carries over faithfully to the finite field context. 
\end{remark}

\subsection{Values of~$\mdf{d}{\FF_{q^2}}$ when $d$ is not a prime power}

\begin{corollary}\label{drown}
Write $d=\prod_{j=1}^n l_j^{k_j}$ as a product of its prime power factors, ordered so that $l_1^{k_1}<l_2^{k_2}<\ldots<l_n^{k_n}$. 
Then there is a set of primes~$p$ of positive Dirichlet density together with integers~$r_p\geq1$ such that for each such pair, writing~$q = p^{r_p}$: 
$$
\mdf{d}{\FF_{q^2}}  \geq  {l_1^{k_1}+1} .
$$
\end{corollary}

\begin{proof}
We need to find a prime power~$q = p^r$ which simultaneously satisfies the hypotheses of proposition~\ref{mdffpp} for every~$l_j^{k_j}$. 
By the Chinese remainder theorem, if we choose a sequence of integers~$a_j \bmod l_j^{k_j}$ for each~$j$ all with the same order --- namely, the value~$2r$ in proposition~\ref{mdffpp} --- then this gives us a unique class~$A \in \left( \ZZ/L\ZZ \right)^\times$, where in principle~$L = \prod_{j}l_j^{k_j}$. 
However, note that if one of the~$l_j$ is~$2$ then we must adjust so that the contribution to~$L$ from the $j$-th prime~$2$ is~$4$,~$a_j = 3$ and~$r=1$ everywhere. 
In this case~$L$ will differ from~$d$. 
By Dirichlet's theorem on primes in arithmetic progressions there are infinitely many primes~$p$ congruent to~$A$ modulo~$L$. 
Such a~$p$ is then by construction congruent to~$a_j \bmod l_j$ for each~$j$. 

Writing $q = p^r$, we now construct for each~$j$ a MUB over~$\FF^{l_j^{k_j}}_{q^2}$ as in the proof of proposition~\ref{mdffpp}. 
We then use the argument in~\cite[lemma~3]{klaro} or~\cite{asch}:  take the tensor product of these component solutions, which yields a set of MUBs in dimension~$d$ over~$\FF_{q^2}$. 
Notice that in this last step we require the condition that the~$r$ be constant, in order to be able to tensor them over the same finite field. 
\end{proof}

\section{Applications of theorem~\ref{nogo}: reduction from $\CC$ to finite fields}
\subsection{Using known solutions in~$\CC$ to search for solutions over~${\FF_{q^2}}$}\label{drownsec}
One way to approach the search for four MUBs in dimension~$6$ is to begin with one of the many known triplets of~MUBs in~$\CC^6$ whose vector entries lie in a complex algebraic number field~$L$ --- see for example~\cite{appingdang} --- and to study its behaviour under reduction at certain admissible primes~$\wp$ of the ring of integers~$\ZZ_L$. 
We then search for more vectors over~${\FF_{q^2}}/{\FF_{q}}$ (where~$q^2 = {\rm N}\wp$) which are MU to them all, and try to lift them back up to the extension~$L_\wp$ of~$\QQ_p$.  
We should point out that we were unable to produce any new evidence this way along the lines of proposition~\ref{quiver}. 

In \cite{dardo}, \cite{ingemar6}, \cite{BWB} etc., many examples are detailed of sets of complex~MUBs in low dimensions: there is also an online classification in~\cite{hadcat}. 
In most cases the entries are given as roots of unity or algebraic numbers of low degree~\cite{appingdang}; hence they may be viewed as lying in some number field~$L$ with a fixed embedding into~$ \CC$. 

Choosing a prime~$\wp \lvert p$ of~$L$ at which the vector entries are $\ZZ_{L_\wp}$-units, the key requirement once again is that the action of the~$q$-th power Frobenius map upon the vector entries reduced modulo~$\wp$, should mimic that of complex conjugation in~$\CC/\RR$.  
For example, it must act via inversion upon any roots of unity not in the base field.  
This is an unavoidable clash with the situation over~$\CC$: for a prime power~$q>3$ the~$(q-1)$-th roots of unity other than~$\pm1$ are not real; whereas they do lie in~${\FF_{q}}$ and so the~$q$-power map fixes them. 
Indeed, the more complicated the Galois action upon the entries, the higher we have to go in general, to find each prime at which Frobenius precisely replicates complex conjugation. 

Locally~$L_\wp$ should be the unramified quadratic extension~$L_\wp / K_\mathfrak{p}$ of a base field~$K_\mathfrak{p}$, writing~$\wp \lvert \mathfrak{p} \lvert p$, such that the action of the non-trivial (Frobenius) element~$\sigma \in \Gal{L_\wp}{K_\mathfrak{p}}$ replicates that of complex conjugation on the $L$-entries of the original matrix. 
Let~$q = N\mathfrak{p}$ denote the absolute (field-theoretic) norm of~$\mathfrak{p}$: so~${\FF_{q}} \cong \ZZ_{K_\mathfrak{p}} / \mathfrak{p}$ is the residue field of the ring of integers~$\ZZ_{K_\mathfrak{p}}$ of~$K_{\mathfrak{p}}$, with~$\FF_{q^2} \cong \ZZ_{L_\wp} / \wp$ the residue field of the ring of integers~$\ZZ_{L_{\wp}}$ of~$L_{\wp}$, and order~2 Galois group~$\Gal{\FF_{q^2}}{\FF_q} \cong \Gal{L_\wp}{K_\mathfrak{p}}$.

\subsection{Examples of MUBs in~$\CC^6$ reduced modulo~$p$: The matrices $H_1, D(0)$ from~\cite[\S4]{dardo}}\label{redux}
We now illustrate this methodology with a couple of simple examples. 
These matrices require the $24$-th roots of unity together with an algebraic number~$b_2 = \frac{-1+2i}{\sqrt{5}}$, where~$i^2 = -1$. 
Note that~$b_2$ is a~$p$-adic unit for all~$p\neq 5$. 
So we need a prime~$p$ which allows~$\mu_{24} \subseteq \FF_{q^2}  \setminus \FF_q$ and which has a square root of~$5$ in~$\FF_q$ but not a square root of~$-1$. 
That is,~$p \equiv -1 \bmod 24$, $\leg{5}{p} = 1$: and~$p \equiv 3 \bmod 4$ automatically by the first condition. 
Of the first few primes~$23,47,71$ satisfying~$p \equiv -1 \bmod 24$, only~$p=71$ contains~$5$ as a quadratic residue. 
For~$q=p=71$ in the notation of \cite{dardo}, the bases were as follows.  
Writing~$\FF_{q^2}^\times = <\gamma> \cong C_{71^2-1}$ and~$u = \gamma^{70}$ for a generator of the subgroup of elements of~$\FF_{q^2}$ of norm~$1$, and~$\delta$ for a fixed choice of element of~$\FF_{q^2}$ of norm~$1/d$ as explained in~\ref{dumbdown}: 

\tiny
\[
H_1 = \delta\left( 
\begin{array}{cccccc}
1 		& 1 		& 1 		& 1 		& 1 		& 1 		\\
u^{54} 	& u^{54} 	& u^6 	& u^{30} 	& u^{30} 	& u^6 	\\
u^{71}	& u^{35} 	& u^{27}	& u^{63} 	& u^{27} 	& u^{63}	\\
u^{54}	& u^{54} 	& u^{30} 	& u^6  	& u^{6} 	& u^{30}	\\
u^{35}	& u^{71} 	& u^{39} 	& u^{51} 	& u^{15}  	& u^{3}	\\
u^{35}	& u^{71} 	& u^{15} 	& u^{3} 	& u^{39} 	& u^{51}
\end{array}
\right)
;\ \ 
D(0) = \delta\left( 
\begin{array}{cccccc}
1 		& 1 		& 1 		& 1 		& 1 		& 1 		\\
1 		& u^{36}	& u^{18} 	& u^{54} 	& u^{54} 	& u^{18} 	\\
1 		& u^{18} 	& u^{36}	& u^{18} 	& u^{54} 	& u^{54}	\\
1 		& u^{54} 	& u^{18}	& u^{36}	& u^{18}	& u^{54}	\\
1 		& u^{54} 	& u^{54} 	& u^{18}	& u^{36}	& u^{18}	\\
1 		& u^{18}	& u^{54} 	& u^{54} 	& u^{18}	& u^{36}
\end{array}
\right).
\]
\normalsize

\section*{Acknowledgements}
Initially phrasing it as a problem on Gr\"obner bases, Al~Kasprzyk and~GM began studying MUBs over finite fields in~2013. 
It quickly became apparent that the only computationally viable approach was an exhaustive search through the vectors themselves. 
This paper is the result of the faster technology now available~\cite{magma}. 
In addition to Al Kasprzyk for helping to get the whole project running, we are grateful to Kevin Buzzard, Ian Grojnowski and Stefan Weigert for helpful conversations during the early stages, and likewise to Markus Grassl. 
A particular vote of thanks is due to Mike Harrison and Ingemar Bengtsson for many key observations, and helpful comments on an earlier draft of this paper. 

GM thanks Myungshik Kim and Terry Rudolph of the QOLS group at Imperial College for their continuing hospitality. 
AT would like to express his gratitude to the \emph{Crankstart Internships} Programme at the University of Oxford for partially funding his work during the project.

\appendix
\section{The searches}\label{sirch}
\subsection{Theorem~\ref{nogo} and proposition~\ref{quiver}}
\subsubsection{Proof of theorem~\ref{nogo}: limits on MUBs in finite fields}\label{tableau}
For a field~$F$ and dimension~$d$, suppose that~$\mdf{d}{F} = n$. 
As a measure of how close we can get to constructing a set of $n+1$ MUBs in $F^d$, we denote by $\nu_dF$ the maximal size, over all sets $\{\mathcal{B}_0,\ldots,\mathcal{B}_{n-1}\}$ of $n$ MUBs in $F^d$, of an orthonormal set of vectors which are mutually unbiased to each~$\mathcal{B}_k$. 
Here are results of the exhaustive searches for the case~$d=6$. 
\small
\begin{table}[h]
\begin{center}
\begin{tabular}{ |c||c c c c c c c c c c c c c c c c|}
\hline
 $q$ &5&7&11&13&17&19&23&25&29&31&37&41&43&47&49&53\\
\hline
$\mdf{6}{\mathbb{F}_{q^2}}$ &3&3&3&2&3&3&3&2&3&3&2&3&3&3&2&3\\
\hline
$\nu_6(\mathbb{F}_{q^2})$ &4&0&0&2&4&0&0&4&4&0&2&4&0&0&2&4\\
\hline
\end{tabular}
\caption{\small Exhaustive search results in 6 dimensions.}\label{table6}
\end{center}
\end{table}\normalsize

\vspace{-5mm}
As stated in~\ref{statknow},~$\mdf{d}{\CC} = d+1$ for prime powers~$d$; and it is a corollary of the upper bound proof in proposition~\ref{finivan} that~$\nu_d(\mathbb{C})=0$. 
Further, in each of the cases $d=2,3,4,5$ it is known that there is precisely one set of $d+1$ MUBs up to equivalence~\cite{BWB}. 
These facts were paralleled in the results of the finite field searches wherever full MUB sets were found to exist.

\small
\begin{table}[htb]
\begin{center}
\begin{tabular}{|c|c|c|c|c|}
\hline
 $d$ & Full $d+1$ MUBs & Partial Sets & Holds For\\
\hline
$2$ & $q \equiv 3 \pmod{4}$ & $\mdf{2}{\mathbb{F}_{q^2}}=2$ for $q \equiv 1 \bmod{4}$ & {\rm all odd}\ $q$ \\
\hline
$3$ & $q \equiv 2 \pmod{3}$ & $\mdf{3}{\mathbb{F}_{q^2}}=1$ for $q \equiv 1 \bmod{3}$ & $(3,q) = 1$ \\
\hline
$4$ & $q \equiv 3 \pmod{4}$ & $\mdf{4}{\mathbb{F}_{q^2}}=3$ for $q \equiv 1 \bmod{4}$ & $q<240$, $q$ {\rm odd} \\
\hline
$5$ & $q \equiv 4 \pmod{5}$ & $\mdf{5}{\mathbb{F}_{2^2}}=4$; $\mdf{5}{\mathbb{F}_{q^2}}=1$ for $q \equiv 1 \bmod 5$; & $q<122, (5,q)=1$ \\
&&  $\mdf{5}{\mathbb{F}_{q^2}}=3$ for $q \equiv  2,3  \bmod 5, q\neq2$ &\\
\hline
\end{tabular}
\caption{\small Exhaustive search results for $d\leq5$.}\label{teatoo}
\end{center}
\end{table}
\normalsize
Table \ref{teatoo} is valid for the ranges of~$q$ shown. 
The assertions for~$d=2,3$ are a straightforward consequence of writing out the matrices with variables, incorporating the protocol detailed in~\ref{rear} and using proposition~\ref{mdffpp}.  
In all cases where $\mdf{d}{\mathbb{F}_{q^2}} > 1$, we find that $\nu_d \mathbb{F}_{q^2} = 0$.

Finally, our limited results for $d=7$ are summarised in the following table. 
\small
\begin{table}[htb]
\begin{center}
\begin{tabular}{ |c||c c c c c c c c c|}
\hline
 $q$ &2&3&4&5&8&9&11&13&17\\
\hline
$\mathcal{M}_7(\mathbb{F}_{q^2})$ &4&8&1&2&4&3&3&8&2\\
\hline
$\nu_7(\mathbb{F}_{q^2})$ &3&0&0&3&3&0&0&0&3\\
\hline
\end{tabular}
\caption{\small Exhaustive search results in 7 dimensions.}\label{sebben}
\end{center}
\end{table}
\normalsize

\subsubsection{Proof of proposition~\ref{quiver}: sets of three hyperbolic MUBs with a fourth set of four MU vectors}\label{quills}
As may be seen in table~\ref{table6}, for each prime~$p \equiv 5 \bmod 12$, we found in dimension~$6$ a set of $3$ MUBs alongside a further set of $4$ orthonormal vectors MU to each MUB. 
The first basis is always taken to be~$\mathcal{B}_0$. 

Lifting the solutions~$p$-adically it became evident that they are perfectly general in characteristic zero: and may be regarded for example as existing in the extension~$K(\sqrt{3}) / K$, where~$K = \QQ(i,\sqrt{6})$ and the Hermitian action is via the map~$\sigma \colon \sqrt{3} \mapsto -\sqrt{3}$. 
When~$q \equiv 5 \bmod 24$,~$\frac{1}{\sqrt{6}}$ is fixed under~$\sigma$; whereas when~$q \equiv 17 \bmod 24$ we need an extra factor of~$i=\sqrt{-1}$ for each vector as in the proof of proposition~\ref{mdffpp}. 
So in particular the following representation reduces well in characteristic~$5\bmod24$: 
$$ 
\tiny
\frac{1}{\sqrt{6}}
\begin{pmatrix}
1	&	1 		&	1 		&	1 		&	1		&	1		\\
1	& -2+\sqrt{3} 	&	-1		&2-\sqrt{3}	&-2+\sqrt{3}	&2-\sqrt{3}	\\
1	&	1 		&-2-\sqrt{3}	&	1		&	1 		&-2+\sqrt{3}	\\
1	& -2-\sqrt{3}	& 2+\sqrt{3} 	&	-1	 	&	1 		&	-1		\\
1	&	1		&-2-\sqrt{3}	&-2+\sqrt{3}	&	1 		&	1		\\
1	&	1 		&2+\sqrt{3}	&	-1		&-2-\sqrt{3} 	&	-1		\\
\end{pmatrix}, $$
$$
\tiny
\frac{1}{\sqrt{6}}
\begin{pmatrix}
1			&	1		&	1 		&	1 		& 1 			& 1			\\
1 			& 	1  		& 	-1 		& -2+\sqrt{3} 	&-1 			& 2-\sqrt{3}	\\
1 			& 	1  		& 	1		& 	1  		& -2-\sqrt{3} 	& -2+\sqrt{3}	\\
1 			&-2-\sqrt{3}	&2+\sqrt{3} 	& -2-\sqrt{3} 	& 2+\sqrt{3} 	& -1			\\
1 			&	1		&-2-\sqrt{3} 	& 	1 		& 1 			& -2+\sqrt{3}	\\
-2-\sqrt{3}		&	1  		&2+\sqrt{3} 	&-2-\sqrt{3} 	& 2+\sqrt{3} 	& -1			\\
\end{pmatrix}.$$
The four vectors, also arranged as columns in a matrix, are: 
$$
\tiny
\frac{1}{\sqrt{6}}
\begin{pmatrix}
1				&	1				&	1 			&	1 				\\
\omega^2(2-\sqrt{3})	&-\omega^2(2-\sqrt{3})	&-\omega(2-\sqrt{3})	&\omega(2-\sqrt{3})		\\
\omega 			&\omega				&\omega^2		&\omega^2			\\
-\omega			&\omega				&\omega^2		&-\omega^2			\\
\omega^2			&\omega^2			&\omega		 	&\omega				\\
-1				&	1  				&1				&-1					\\
\end{pmatrix},$$
\normalsize
where~$\omega$ is a primitive cube root of unity. 
Note that~$\sigma \colon \omega \mapsto \omega^2$.

\subsection{The search algorithms}\label{simp}
The exhaustive computer algorithm search routine is essentially self-explanatory. 
We just give a few basic background facts. 
The starting point in our sets of MUBs is always~$\mathcal{B}_0=\{\mathbf{e}_1,\ldots,\mathbf{e}_d\}$, represented by the identity matrix. 

\subsubsection{Formal equivalence of notions of MU in $\CC$ and $\FF_{q^2}$}\label{axcee}
Fix an odd prime~$p$. 
The definition in~(\ref{mubness}) is formally aligned with that in the complex case: the square norm of the `absolute value' of the Hermitian inner product of each pair of vectors from distinct bases must equal~$1/d$.  
Hence our stipulation that~$(d,p)=1$. 
It is also somewhat meaningless to have~$d$ reinterpreted modulo the characteristic~$p$. 
For example, if~$p=5$ and~$d=6$, then every unit vector is `MUB to itself'. 
So ideally, we insist further that~$d<p$. 

Fix~$q=p^r$ for some~$r\geq1$: we always work over the extension ${\FF_{q^2}} / {\FF_{q}}$ with Galois group generated by the Frobenius element~$\sigma:x\mapsto x^q$ of order~$2$. 
Since the quadratic extension of a finite field is unique up to isomorphism, we are assured of the existence in~${\FF_{q^2}}$, if not already in~${\FF_{q}}$, of a square root of~$d$.

As an aside, however, we illustrate a major difference between our situation --- where we are in the end searching for a solution in a number field --- and those in~$\RR$ or~$\CC$ where the base field always contains~$\sqrt{d}$. 
Let~$d>1$ be a non-square positive integer. 
Consider the corresponding Hermitian structure which arises via the Galois group~$\Gal{\QQ(\sqrt{d})}{\QQ}$.   
Whereas in~$\RR$ or~$\CC$ we would have~$N(1/\sqrt{d}) = (1/\sqrt{d})^2 = 1/d$, here instead~$N(1/\sqrt{d}) = (1/\sqrt{d})(-1/\sqrt{d}) = -1/d$. 
This is the context in which the \emph{hyperbolic MUBs} arise in proposition~\ref{quiver} and~\ref{quills}.

\subsubsection{Hadamard matrices}\label{axceetoo}
Now let~$F/K$ be any quadratic extension of fields and let~$V$ be a {unitary space}. 
Given any matrix operator~$M$ on~$V$, let~$M^\dag$ denote the Hermitian conjugate transpose of~$M$: so in particular~$\mathcal{B}$ represents an orthonormal basis if and only if its associated matrix~$\mathcal{B}$ is unitary; which in turn is true iff~$ \mathcal{B}^\dag \mathcal{B} = \mathcal{B} \mathcal{B}^\dag = \mathcal{B}_0$. 
In view of the fact that much of the literature on this subject is phrased in these terms --- see for example~\cite{ingemar6, hadcat, dardo} --- we note that a matrix~$\mathcal{B}$ which is MU to~$\mathcal{B}_0$ is often called an~\emph{$F$-Hadamard matrix} --- or where the context is clear just a~\emph{Hadamard matrix} --- drawing upon the terminology in the complex case\footnote{Using the convention of Bengtsson~\emph{et al} in~\cite{ingemar6}, in that Hadamard matrices are~\emph{unitary}, as opposed to~\emph{unimodular} as in say~\cite{hadcat}.  
These differ merely by a one-off normalisation by~$\tfrac{1}{\sqrt{d}}$.}.   

A second basis~${\mathcal{C}}$ represented by another~$F$-Hadamard matrix is then in turn MU to~${\mathcal{B}}$ iff the ordinary matrix product~${{\mathcal{B}}}^\dag {{\mathcal{C}}}$ is also itself an~$F$-Hadamard matrix. 
Indeed, the condition in definition~\ref{mubdef} is identical to requiring that every entry~${ \langle \mathbf{b}_i , \mathbf{c}_j \rangle }$ of the unitary matrix~${\mathcal{B}}^\dag {\mathcal{C}}$ of pair-wise Hermitian dot products have that norm. 
Since for any unitary matrix~$\mathcal{B}$ and any matrices $\mathcal{C}$, $\mathcal{D}$ it is the case that $(\mathcal{B}^\dag\mathcal{C})^\dag (\mathcal{B}^\dag\mathcal{D})  =  \mathcal{C}^\dag\mathcal{D}$, including~$\mathcal{B}_0$ in any set of MUBs as the first member is no restriction. 
This then forces all of the entries~${{b}_j} \in F$ of all other basis 
vectors~$\mathbf{b}$ of another MUB~${\mathcal{B}}$, say, to satisfy~$\norm{F}{K}{b_j}  =  1/d$. 
However, although when~${\mathcal{B}}$ and~${\mathcal{C}}$ are MU to one another, the basis associated to~${\mathcal{D}}  =  {{\mathcal{B}}^\dag{\mathcal{C}}}$ is by definition itself MU with respect to~$\mathcal{B}_0$, in most cases~${\mathcal{D}}$ will be MU neither to~${\mathcal{B}}$ nor to~${\mathcal{C}}$.  
Parenthetically, the order in which we take the Hermitian product is important: when~${{\mathcal{B}}^\dag{\mathcal{C}}}$ satisfies our conditions, this nevertheless reveals nothing in particular about the unitary~$ {{\mathcal{B}}{\mathcal{C}}}^\dag$.

\subsubsection{Transformations which preserve the MU properties: Rearrangements of the columns and rows of the MUB matrices}\cite[\S III]{ingemar6}, \cite{debz}.  \label{rear}
Let $\mathcal{S} = \{ \mathcal{B}_1 , \ldots , \mathcal{B}_n \}$ be a set of MUBs in~$F^d$ which is assumed not to include~$\mathcal{B}_0$, but to be MU to it. 
Given any basis matrix~$\mathcal{B}_k \in \mathcal{S}$ we are free to rearrange the columns in any order without affecting the MU properties. 
This amounts to right-multiplying~$\mathcal{B}_k$ by a~$d\times d$ permutation matrix. 
Similarly and separately, we may rearrange the rows of every~$\mathcal{B}_k$ for~$1\le k \le n$,~\emph{provided that we do an identical rearrangement on all~$n$ matrices simultaneously}. 
This may be effected by left-multiplying them all by \emph{the same} permutation matrix. 

By analogy with the real case, two pairs of vectors ${\mathbf{u}},{\mathbf{v}}$ and~${\mathbf{u'}},{\mathbf{v'}}$ subtend the same `angle' if and only if~$\norm{F}{K}{\langle \mathbf{u} , \mathbf{v} \rangle}  =  \norm{F}{K}{\langle \mathbf{u'} , \mathbf{v'} \rangle}$. 
This angle is unchanged if we multiply~${\mathbf{u}}$ or~${\mathbf{v}}$ or both by elements of~$F$ of norm~1. 
Regarding each~$\mathcal{B}_k$ separately therefore we are free to right-multiply any or all of them by (possibly) different diagonal matrices of norm~1 elements. 
So in particular, we may assume without loss of generality that the first entry in every vector of every one of the~$\mathcal{B}_k$ is equal to the same value, which we choose to be the quantity~$\delta$ below, where~$\norm{F}{K}{\delta} = 1/d$. 

Equally, we may multiply the rows of the bases by norm~$1$ elements, again provided that we use the same element on each corresponding row of each matrix at the same time. 
This is achieved by left-multiplying every one of the~$\mathcal{B}_k$ by the same diagonal matrix with norm~$1$ diagonal entries. 
In this way, we are able for example to ensure that the first vector of the first basis~$\mathcal{B}_1$ should have all of its entries equal to~$\delta$.  
This was illustrated in the examples in~\ref{drownsec} and~\ref{quills}. 

\subsubsection{Transformations which preserve the MU properties: Reducing the search space from $d$ dimensions to ${d-1}$}\label{dumbdown}
We make a few obvious simplifications. 
The norm homomorphism~$N$ on multiplicative groups of finite fields is surjective, yielding the following short exact sequence of finite groups defining the kernel~${U}$:
$$
1\lra {U}\lra {\FF^\times_{q^2}}\overset{N}\lra{\FF^\times_{q}}\lra1.
$$
If~$\gamma$ is any generator for the multiplicative subgroup~${\FF^\times_{q^2}}$ of~${\FF_{q^2}}$ then~${U} =\ <\gamma^{q-1}>$ is the unique subgroup of~${\FF^\times_{q^2}}$ of index~$q-1$ (equivalently, of order~$q+1$).   
Let~$\delta \in {\FF^\times_{q^2}}$ lie above~$d^{-1}$; then~$\delta U$ is the coset containing all elements of norm~$d^{-1} = N\delta = \delta^{q+1}$. 

By setting the first basis to be~$\mathcal{B}_0$, we may restrict our attention to vectors whose entries lie in~${\delta U}$.   
Indeed we may view~$\left({\delta U}\right)^d=\delta {U}^d$, the cartesian product of~$d$ copies of~${\delta U}$, as a sub\emph{set} of~${\FF^d_{q^2}}$ and focus our search within it.  
In accordance with the previous section, we are free to choose the first vector of our first new basis~$\mathcal{B}_1$ to consist of all~$\delta$'s; and moreover the first element of each of our non-$\mathcal{B}_0$ vectors also may be chosen to be~$\delta$.      

Just as with the situation over the complex extension of the reals, for each vector entry we require two variables over~$\FF_q$ to constitute just one over~$\FF_{q^2}$, so that the Frobenius map may be invoked separately all the way up the~$p$-adic tower. 
So we need ten variables in~$\FF_q$ for each vector in our search.

\section{The equations and reduction mod~$p$}
We give a very brief outline of some facts which arise when viewing this problem through the lens of algebraic geometry. 
Throughout we assume we are in the context of a unitary space over a field~$F$, where~$F/K$ is a quadratic field extension as in the text. 

In order for a basis~$\mathcal{C}$ to be MU to an ON basis~$\mathcal{B}$, it must satisfy three separate uniform sets of highly symmetric equations: (I) equality to~$\frac{1}{d}$ of the norm of each vector entry, to prove that the vectors are MU to the computational basis~$\mathcal{B}_0$ --- which also forces them to be unit vectors; (II) orthogonality among themselves; and (III) mutual unbiasedness to~$\mathcal{B}$.

\subsection{The individual defining polynomials reduced modulo a prime}
Zauner's conjecture predicts that the ideal~$\mathcal{J}$ of the polynomial ring~$\ZZ[\underline{X}]$ 
generated by the~$261$ multivariate quartic and quadratic polynomials in~$N=216$ real variables defining four MUBs in~$\CC^6$, should be the whole ring~$\ZZ[\underline{X}]$. 
This is a Gr\"obner basis calculation way beyond the power of any known algorithm. 
In particular, good reduction modulo a prime is unprovable, other than by inspection for the primes~$2$ and~$3$. 
So although the following result puts the reduction behaviour of the \emph{individual} MUB defining polynomials into context, it nevertheless says nothing about the multiplicity of reduced solutions in their intersection varieties. 
See~\ref{eqnshape} for the actual polynomials for~$d=2$: the general case is entirely analogous. 

\begin{proposition}
The defining polynomials for a set of four MUBs over the complex numbers in dimension~$6$ are each absolutely irreducible. 
Consequently each polynomial is guaranteed to have good reduction outside a finite set of primes. 
\end{proposition}

\begin{proof}
Let~$f(X)$ be one of the defining polynomials, where~$X = \underline{X}$ denotes a tuple of variables. 
Choosing some high enough prime~$p>3$ we verified using MAGMA the irreducibility of~$f$ over~$\FF_p$. 
Examples of \emph{simple zeroes} for~$f$ --- that is, solutions with entries in~$\FF_p$ where the Jacobian does not vanish completely --- are then straightforward to construct by hand. 
Using~\cite[thm1]{ragout} it follows therefore that~$f(X)$ is \emph{absolutely irreducible}: that is, it remains irreducible over the algebraic closure~$\overline{\FF_p}$ of~$\FF_p$. 

Finally by theorem~2 of the same paper, since the degree of~$f$ is unchanged under reduction modulo~$p$, we may lift that absolute irreducibility to~$\QQ$. 
That is to say, it is irreducible over~$\overline{\QQ}$. 
The second assertion then follows from Emmy Noether's irreducibility criteria~\cite[\S V]{schmidt}. 
\end{proof}

The known theoretical upper bounds for these bad primes are huge~\cite{ragout}; however in practice the set of bad primes for any known MUB solution is tiny.

\subsection{The shape of the equations in dimension~$d=2$}\label{eqnshape}
We give a representative example only for dimension~$d=2$ for the sake of brevity; the systems scale to higher dimensions in an entirely predictable and uniform way. 
The notation for the vector entries only applies to this section. 

Other than the standard assumption that the first basis of any set always be~$\mathcal{B}_0$, we leave them in their most general format so as to illustrate the inherent symmetries. 
Representing the MUB vectors as usual by columns of a~$d\times d$ matrix, we may establish the system as follows, using~$i$ as a generic symbol for a generator of the quadratic field extension~$F/K$: 
\small
\[
{\mathcal{B}_0} =  \left( 
\begin{array}{cc}
1 		& 0 			\\
0 		& 1 		
\end{array}
\right) ;\ \ 
{\mathcal{B}_1} =  \left( 
\begin{array}{cc}
e+i\epsilon& 	g +i\gamma\\
f + i\phi	& h +i\chi		
\end{array}
\right) ,\ \ 
{\mathcal{B}_2} =  \left( 
\begin{array}{cc}
s+i\sigma		& u +i\mu			\\
t + i\tau		& v +i\nu		
\end{array}
\right) .
\]
\normalsize
We are forced to have two $K$-valued variables per~$F$-valued vector entry, per basis vector: since --- with an eye to lifting eventually to characteristic zero --- we must explicitly build the Hermitian conjugate into the equations. 
Prior to any adjustments along the lines of those in~\ref{rear}, the resulting equations are as follows. 

The roman numerals refer to the introduction to this appendix. 
(I) First, we must ensure that all entries in all vectors have norm~$\frac{1}{d}$, to ensure MUB-ness with the computational basis~$\mathcal{B}_0$. 
This also forces them to be unit vectors. 
In principle this should give~$d$ inhomogeneous quadratic equations per vector, hence~$d^2$ per basis beyond~$\mathcal{B}_0$, viz.: 
\begin{eqnarray*}
\small{
e^2 + \epsilon^2 - 1/2 \hbox{\ , \ \ }
f^2 + \phi^2 - 1/2 \hbox{\ , \ \ }
g^2 + \gamma^2 - 1/2  \hbox{\ , \ \ }
h^2 + \chi^2 - 1/2  \hbox{\ , \ \ } }\\
\small{
s^2 + \sigma^2 - 1/2  \hbox{\ , \ \ }
t^2 + \tau^2 - 1/2  \hbox{\ , \ \ }
u^2 + \mu^2 - 1/2  \hbox{\ , \ \ }
v^2 + \nu^2 - 1/2 }\ . 
\end{eqnarray*}
Note that this places the candidates for MUB vector entries upon a~$2d^2$-torus.

(II) Secondly, each basis must be orthonormal. 
By~(I) they are unit vectors, so it remains to check orthogonality. 
There are~$\binom{d}{2}$ comparisons to be made, each with real and imaginary parts, yielding a total of~$d(d-1)$ extra homogeneous quadratic equations per basis; in our case: 
\begin{eqnarray*}
\small{
e g + \epsilon \gamma + f h + \phi \chi \hbox{\ , \ \ }
-e \gamma + \epsilon g - f \chi + \phi h \hbox{\ , \ \ } } \\
\small{
s u + \sigma \mu + t v + \tau \nu \hbox{\ , \ \ }
-s \mu + \sigma u - t \nu + \tau v }.
\end{eqnarray*}
(III) Finally the actual MUB-ness comparisons require a recursive structure of new equations which --- in taking norms of sums of Hermitian products --- gives~$d^2$ more inhomogeneous quartic equations per (unordered) \emph{pair of bases}: 
\small
\begin{eqnarray*}
\small{
e^2 s^2 + e^2 \sigma^2 + 2 e f s t + 2 e f \sigma \tau + 2 e \phi s \tau - 2 e \phi \sigma t + \epsilon^2 s^2 + \epsilon^2 \sigma^2 - 2 \epsilon f s \tau}&&\\ \small{+ 2 \epsilon f \sigma t + 2 \epsilon \phi s t + 2 \epsilon \phi \sigma \tau + f^2 t^2 + f^2 \tau^2 + \phi^2 t^2 + \phi^2 \tau^2 - 2 \ ,} \\
\small{e^2 u^2 + e^2 \mu^2 + 2 e f u v + 2 e f \mu \nu + 2 e \phi u \nu - 2 e \phi \mu v + \epsilon^2 u^2 + \epsilon^2 \mu^2 - 2 \epsilon f u \nu + 2 \epsilon f \mu v }&&\\ \small{+ 2 \epsilon \phi u v + 2 \epsilon \phi \mu \nu + f^2 v^2 + f^2 \nu^2 + \phi^2 v^2 + \phi^2 \nu^2 - 2 \ ,} \\
\small{ g^2 s^2 + g^2 \sigma^2 + 2 g h s t + 2 g h \sigma \tau + 2 g \chi s \tau - 2 g \chi \sigma t + \gamma^2 s^2 + \gamma^2 \sigma^2 - 2 \gamma h s \tau}&&\\ \small{ + 2 \gamma h \sigma t + 2 \gamma \chi s t + 2 \gamma \chi \sigma \tau + h^2 t^2 + h^2 \tau^2 + \chi^2 t^2 + \chi^2 \tau^2 - 2 \ ,} \\
\small{ g^2 u^2 + g^2 \mu^2 + 2 g h u v + 2 g h \mu \nu + 2 g \chi u \nu - 2 g \chi \mu v + \gamma^2 u^2 + \gamma^2 \mu^2 - 2 \gamma h u \nu + 2 \gamma h \mu v}&&\\ \small{ + 2 \gamma \chi u v + 2 \gamma \chi \mu \nu + h^2 v^2 + h^2 \nu^2 + \chi^2 v^2 + \chi^2 \nu^2 - 2\ . }
\end{eqnarray*}
\normalsize

\input{FF_MUBs_bibzz}
\end{document}

%% file: defns.tex
\usepackage{hyperref}
\usepackage{amssymb}
\usepackage{calligra}
\usepackage{amsfonts}

\usepackage{enumerate}
\usepackage{comment}

\newtheorem{theorem}{Theorem}[section]
\newtheorem{proposition}[theorem]{Proposition}

\newtheorem*{remark}{Remark}
\newtheorem{corollary}[theorem]{Corollary}

\newtheorem{definition}{Definition}
\newtheorem*{theorem*}{Theorem}

\def\ZZ{\mathbb{Z}} 
\def\RR{\mathbb{R}} 
\def\CC{\mathbb{C}} 
\def\NN{\mathbb{N}} 
\def\QQ{\mathbb{Q}}
\def\FF{\mathbb{F}}
\def\II{\mathbb{I}}

\newcommand{\leg}[2]{({\scriptstyle{\frac{#1}{#2}}})}

\def\lra{{\ \xrightarrow{\ \ \ }\ }}

\newcommand{\Gal}[2]{\hbox{\small{\rm{Gal}}}_{#1/#2}}

\newcommand{\norm}[3]{{\hbox{\small{\rm{N}}}_{\scriptscriptstyle{#1}/\scriptscriptstyle{#2}}}{\hbox{${#3}$}}}